\documentclass[runningheads]{llncs}
\usepackage[T1]{fontenc}
\usepackage{graphicx}
\usepackage{bm}
\usepackage{amsmath}
\usepackage{float}
\usepackage{diagbox}
\usepackage{amsfonts}
\usepackage{threeparttable}
\usepackage{tabularx}
\usepackage{xcolor}
\usepackage{cleveref}
\usepackage{listings}

\begin{document}
\title{Discrete Effort Distribution via \\Regret-enabled Greedy Algorithm 
\thanks{This research is supported by Department of Science and Technology of Guangdong Province (Project No. 2021QN02X239) and Shenzhen Science and Technology Program (Grant No. 202206193000001, 20220817175048002).}}
%
%
\author{Song Cao\orcidID{0009-0002-1760-3820},
Taikun Zhu\orcidID{0000-0001-7365-9576}, \\ and 
Kai Jin\orcidID{0000-0003-3720-5117} \thanks{Corresponding author: Kai Jin. \url{cscjjk@gmail.com}.}}
%
%
\institute{Shenzhen Campus of Sun Yat-sen University, Shenzhen, Guangdong, China
\email{caos6@mail2.sysu.edu.cn}\\
\email{zhutk3@mail2.sysu.edu.cn}\\
\email{jink8@mail.sysu.edu.cn}\\
}

\authorrunning{Cao \emph{et al.}}

\maketitle              
\begin{abstract}
This paper addresses resource allocation problem with a separable objective function under a single linear constraint, formulated as maximizing $\sum_{j=1}^{n}R_j(x_j)$ subject to $\sum_{j=1}^{n}x_j=k$ and $x_j\in\{0,\dots,m\}$. While classical dynamic programming approach solves this problem in $O(n^2m^2)$ time, we propose a regret-enabled greedy algorithm that achieves $O(n\log n)$ time when $m=O(1)$. The algorithm significantly outperforms traditional dynamic programming for small $m$. Our algorithm actually solves the problem for all $k~(0\leq k\leq nm)$ in the mentioned time.
 
\end{abstract}

\section{Introduction}
We consider the effort distribution problem with separable objective function and one linear constraint. It can be formulated as follows.
\begin{align*}
  &\text{Maximize: }\sum_{j=1}^{n}R_j(x_j),\\
  &\text{subject to: }\sum_{j=1}^{n}x_j=k,\indent x_j\in\{0,...,m\}, m>0 \text{ and }1\leq k\leq nm,
\end{align*}
where,

$n$ = the number of projects,

$k$ = the total number of efforts,

$m$ = the maximal number of efforts allowed to be allocated to each project,

$x_j$ = the number of efforts allocated to $j$-th project, and

$R_j(x_j)$ = the revenue $j$-th project generates when it is allocated $x_j$ efforts.

Assume that $R_j(0)=0$.

It can be solved by dynamic programming in $O(n^2m^2)$ time. Let $dp[j,k]$ be maximal revenue for the first $j$ projects with $k$ efforts. Then we 
have
$$dp[j,k]=\max_{0\leq x_j\leq m} \left\{ dp[j-1,k-x_j]+R_j(x_j) \right\}.$$

While prior works imposes concavity or near concavity constraints on $R_j$, our approach removes these constraints, requiring only separability of $R_j$. In this paper, we give an $O(n\log n)$ time algorithm based on a regret-enabled greedy framework, which solves the problem for all $k$ ($0\leq k\leq nm$) when $m=O(1)$. 

Traditional greedy algorithms iteratively make locally optimal decisions to achieve global optimal solution, but they fail in our context. To overcome this we use a regrettable greedy mechanism---a paradigm that allows strategic revocation of prior decisions. Specifically, our algorithm adjusts allocations by (1) removing $t$ (a parameter dependent on $m$) efforts from some projects, and (2) allocating $t+1$ efforts to some projects to maximize incremental revenue at each step.
\medskip

The organization of this paper is as follows. 
In Section~\ref{sect:greedy-for-m}, we establish a crucial property of optimal solutions, and present our main algorithm in Section~\ref{sect:mainalgorithm}. Furthermore, for the special case where all revenue functions $R_j$ are convex, we introduce two algorithms in Section~\ref{sect:convex}: one computes optimal solutions for all $k$ in $O(nm+n\log n)$ time, while the other runs in $O(nm)$ time for a given $k$.

 In Section~\ref{sect:faster}, we define a class of functions called \emph{oscillating concave functions} and demonstrate a computational property: if $f$ is concave and $g$ is oscillating concave, their (max,+) convolution can be computed in $O(n)$ time. Based on this property, we describe an algorithm for $m=2$ that achieves $O(n)$ time after an initial sorting step.

\newcommand{\x}{\mathbf{x}}
\newcommand{\y}{\mathbf{y}}
\newcommand{\calP}{\mathcal{P}}

\begin{definition}
A distribution of $k$ efforts to the $n$ projects can be described by a vector
  $\x = (x_1,\ldots,x_n)$ where $x_i \in \{0,\ldots,m\}$ and $\sum_i x_i=k$.
Such a vector is called a \emph{$k$-profile}.
For $k\geq 0$, denote by $\calP_k$ the set of $k$-profiles.

A $k$-profile is optimal if its revenue $\sum_i R_i(x_i)$ is the largest
  among $\calP_k$.
\end{definition}

\subsection{Related works}\label{subsect:related}

The problem we studied falls under the broad category of resource allocation. Resource allocation problem involves determining the cost-optimal distribution of constrained resources among competing activities under fixed resource availability. The multi-objective resource allocation problem (MORAP) has been formally characterized through network flow modeling by Osman \emph{et al.} \cite{Osman}, establishing a generalized framework for handling different optimization criteria under resource constraints. 
Beyond that, resource allocation problem widely appears in manufacturing, computing, finance and network communication. 
Bitran and Tirupati \cite{Bitran} formulated two nonlinear resource allocation problems---targeting problem (TP) and balancing problem (BP)---for multi-product manufacturing systems. Bitran and Saarkar \cite{Bitran2} later proposed an exact iterative algorithm for TP. Rajkumar \emph{et al.} \cite{Rajkumar} presented an analytical model to measure quality of service (QoS) management, which referred to as QoS-based Resource Allocation Model (Q-RAM). Bretthauer \emph{et al.} \cite{Bretthauer2} transferred various versions of stratified random sampling plan problem into resource allocation problems with convex objective and linear constraint. And they provided two branch-and-bound algorithms to solve these problems.

A fundamental variant known as the \emph{simple resource allocation problem} \cite{Katoh} involves minimizing separable convex objective functions (or maximizing separable concave objective functions) with a single linear constraint, solvable via classical greedy algorithms \cite{Fox,Shih}. Subsequent research has extended this framework along two directions: generalizing objective functions and complex constraints. For instance, Federgruen and Groenevelt \cite{Federgruen} developed greedy algorithms for weakly concave objectives, while Murota \cite{Murota} introduced M-convex functions—a specialized subclass of convex functions later studied by Shioura \cite{Shioura} for polynomial-time minimization. Nonlinear constraints were addressed by Bretthauder and Shetty \cite{Bretthauder}, who proposed a branch-and-bound algorithm for separable concave objectives. Multi-objective scenarios were explored by Osman \emph{et al.}\cite{Osman} using genetic algorithms, and online stochastic settings were investigated by Devanur \emph{et al.} \cite{Devanur} through a distributional model yielding an $1-O(\epsilon)$-approximation algorithm. Recent work by Deng \emph{et al.} \cite{Deng} further extended the framework to nonsmooth objectives under weight-balanced digraph constraints via distributed continuous-time methods. In this paper, we focus on generalizing the objective function by removing its concavity constraint.


From a computational perspective, our problem admits the computation of (max,+) convolution. Given two sequences $\{x_i\}_{i=1}^n$ and $\{y_i\}_{i=1}^n$, their (max,+) convolution computes $z_k=\max_{i=0}^k(x_i+y_{k-i})$, with (min,+) convolution defined analogously. Many problems occur to be computation of such convolution, such as the Tree Sparsity problem and Knapsack problem. 
 
While naively computable in $O(n^2)$ time, Cygan \emph{et al.}\cite{Cygan} put forward that there is no $O(n^{2-\epsilon})$ algorithm where $\epsilon>0$ for (min,+) convolution. Subsequent improvements include Bremner \emph{et al.}'s $O(n^2 / \lg n)$ algorithm \cite{Bremner} and Bussieck \emph{et al.}'s $O(n\log n)$ expected-time algorithm for random inputs \cite{Bussieck}. Special sequence structures enable faster computation: 
When $x$ and $y$ are both convex, their (min,+) convolution can be easily computed in $O(n)$ time. For monotone integer sequences bounded by $O(n)$, Chan \emph{et al.} \cite{Chan} achieved $O(n^{1.859})$, later refined to $\tilde{O}(n^{1.5})$ upper bound by Chi \emph{et al.} \cite{Chi}. Bringmann \cite{Bringmann} further considered $\Delta$-near convex functions-those approximable by convex functions within additive error $\Delta$-yielding an $\tilde{O}(n\Delta)$ algorithm. The conclusion we obtain in Section~\ref{sect:faster} slightly broadens the class of functions for which (max,+) convolution can be computed in $O(n)$ time.

Our resource allocation problem is closely related to the subset-sum and Knapsack problem \cite{Knapsack1,Bringmann-BKP}. 
Let $W$ denote the maximum weight of the items, and $P$ denote the maximum profit of the items.
  Pisinger \cite{Knapsack1} shows that (1) the subset-sum problem can be solved in $O(nW)$ time, improving over the trivial $O(n^2W)$ bound,
     and (2) the Knapsack problem can be solved in $O(nWP)$ time.
  Recently, an $\tilde{O}(n+W^2)$ time algorithm is given for the Knapsack problem by Bringmann \emph{et al.} \cite{Bringmann-BKP}.
  See more related work of the Knapsack problem (with the parameter $W$) in \cite{Bringmann-BKP}.
  Note that the Knapsack problem is a special case of (and hence easier than) our resource allocation problem.
    An item with weight $w$ can be seen as a project $j$; moreover, $R_j(w)$ is the profit of this item,
      where $R_j(x_j)=-\infty$ for $x_j\neq w$. Be aware that $m=W$ is the maximum weight of the items.
  

\subsection{Preliminaries: some observations on multisets}

For convenience, in this paper a multiset refers to a multiset of $[m]=\{1,\ldots,m\}$.

A pair of multisets $(A,B)$ is \emph{reducible} if 
     the sum of a nonempty subset of $A$ equals the sum of a nonempty subset of $B$,
and is \emph{irreducible} otherwise.

\begin{example}
Reducible: $(A,B)=(\{1,2,2,2\},\{3,3\})$, $(A,B)=(\{1,3\},\{2,2\})$.

Irreducible: $(A,B)=(\{2,2\},\{3\})$, $(A,B)=(\{3\},\{1,1\})$.
\end{example}

\medskip Denote $\lambda_m=m^2.$

For any multiset $A$, its sum of elements is denoted as $\sum A$.

\begin{lemma}\label{lemma:property-for-m}
A pair of multisets $(A,B)$ is reducible if $\sum A\geq \lambda_m$ and $\sum B\geq \lambda_m$.
\end{lemma} 

\begin{proof}
We first prove an observation: A pair of multisets $(A,B)$ is reducible if $A,B$ each have $m$ elements in $[m]$. 

Without loss of generality, suppose $\sum A\leq \sum B$. For convenience, let $a[i] (i\in[m])$ denote the sum of first $i$ elements in $A$, and $b[j] (j\in[m])$ denote the sum of first $j$ elements in $B$. Notice that $a[i]$ and $b[j]$ are both strictly increasing sequences.

For each $i\in [m]$, let
$$
c[i]=\left\{
                    \begin{array}{ll}
                      a[i]-b[j^*], & \text{the largest }j^*\text{ satisfying }a[i]\geq b[j^*]; \\
                      a[i], & \text{no such }j^*\text{ exists}.
                    \end{array}
                  \right.
$$
We claim that $c[i]\leq m-1$, the proof is as follows. \\
(1) If $c[i]=a[i]-b[j^*]$, and $j^*\ne m$. We have $b[j^*]\leq a[i]<b[j^*+1]$ and $b[j^*+1]-b[j^*]\leq m$, therefore $a[i]-b[j^*]\leq m-1$. \\
(2) If $c[i]=a[i]-b[m]$. Suppose $a[i]-b[m]\geq m$, we have
$\sum A-\sum B\geq a[i]-\sum B=a[i]-b[m]\geq m$, which conflicts to $\sum A\leq \sum B$. \\
(3) If $c[i]=a[i]$. By the definition of $c[i]$ we know $a[i]<b[1]\leq m$.

\smallskip
If there exists $c[i_0]=0$, then $a[i_0]=b[j^*]$, which means $(A,B)$ is reducible. Otherwise, we have $1\leq c[i]\leq m-1$ for each $i\in [m]$. And by Pigeonhole Principle, there exist $c[i_1]=c[i_2]$ ($i_1<i_2$), which means one of the following holds: (1) $a[i_1]-b[j_1^*]=a[i_2]-b[j_2^*]$; (2) $a[i_1]=a[i_2]-b[j_2^*]$. Each of them can demonstrate that $(A,B)$ is reducible.

\bigskip
Finally we go back to Lemma~\ref{lemma:property-for-m}. Suppose $\sum A\geq m^2$, then $A$ has at least $m$ elements (otherwise $\sum A\leq (m-1)m$). Similarly, suppose $\sum B\geq m^2$, then $B$ has at least $m$ elements. By the observation above, $(A,B)$ is reducible. \qed

%
%
%
%
%
%
%
%
%
%
\end{proof}

\begin{remark}\label{remark:conjecture}
Bringmann \emph{et al.} \cite{Bringmann-BKP} gave another result with significantly increased analytical complexity: 

\begin{lemma}\cite{Bringmann-BKP}
A pair of multisets $(A,B)$ is reducible if 
$$|A|\geq 1500\left(\log^3 (2|A|)\mu(A)m\right)^{1/2}$$ 
and 
$$\textstyle\sum B\geq 340000\log(2|A|)\mu(A)m^2/|A|,$$
where $\mu(A)$ denotes the maximal multiplicity of elements in $A$.
\end{lemma} 
\end{remark}

\subsection{Irreducible pair $(A,B)$ with $\sum A-\sum B=1$}\label{subsect:ip-brute-force}

Suppose we want to enumerate irreducible pairs $(A,B)$ satisfying $\sum A - \sum B= 1$ (for some fixed small $m$) (which will be used in our algorithm).
    We only need to focus on $(A,B)$ with $\sum B<\lambda_m$
   (since otherwise $\sum A, \sum B \geq \lambda_m$, and $(A,B)$ must be reducible by \Cref{lemma:property-for-m}).
Therefore we can enumerate all target pairs by brute-force programs
  (check all $(A,B)$ where $\sum B<\lambda_m$ and $\sum A=\sum B + 1$).

\begin{example}\label{example:ip2}
All irreducible pairs with $\sum A -\sum B=1$ for $m=2$ are:
\begin{equation*}
\begin{aligned}
    A=&\{1\}, &B=& \emptyset;\\
    A=&\{2\}, &B=& \{1\}.
\end{aligned}
\end{equation*}
\end{example}

\begin{example}\label{example:ip3}
All irreducible pairs with $\sum A -\sum B=1$ for $m=3$ are:
\begin{equation*}
\begin{aligned}
    A=&\{1\}, &B=& \emptyset;\\
    A=&\{2\}, &B=& \{1\};\\
    A=&\{3\}, &B=& \{2\};\\
    A=&\{3\}, &B=& \{1,1\};\\
    A=&\{2,2\}, &B=& \{3\}.
\end{aligned}
\end{equation*}
\end{example}

The number of irreducible pairs with $\sum A-\sum B=1$ will be denoted by $p_m$, or $p$ for simplicity.
According to our brute-force programs, 
$$p_1=1,~p_2=2,~p_3=5,~p_4=11,~p_5=27.$$

\section{A crucial property of the optimal $k$-profiles}\label{sect:greedy-for-m}

\newcommand{\diff}{\mathsf{diff}}

\begin{definition}
Assume $\x=(x_1,...,x_n)\in \calP_{k}$ and $\y=(y_1,...,y_n)\in \calP_{k+1}$.

Define $\diff(\x,\y)=(A,B)$ and call it the difference of $(\x,\y)$, where 
\begin{eqnarray}
A&=&\{y_i-x_i \mid i\in [n] \text{ and } y_i>x_i\}. \label{eqn:diff1} \\
B&=&\{x_i-y_i \mid i\in [n] \text{ and } x_i>y_i\}. \label{eqn:diff2}
\end{eqnarray}

Notice that $\sum A - \sum B = \sum_i (y_i-x_i) = (k+1)-(k) = 1$.
\end{definition}

\begin{example}
$\diff((1,1),(3,0))=(\{2\},\{1\})$. $\diff((2,2,2),(3,1,3))=(\{1,1\},\{1\})$.
\end{example}

\begin{lemma}\label{lemma:diff-irreducible}
For any optimal $k$-profile $\x$, where $k<nm$,
  there exists an optimal $(k+1)$-profile $\y$ such that $\diff(\x,\y)$ is irreducible.
\end{lemma}

\begin{proof}
    First of all, take any $(k+1)$-optimal profile $\y$.
    If $\diff(\x,\y)$ is irreducible, we are done.
    Now, suppose to the opposite that $(A,B)=\diff(\x,\y)$ is reducible. 

For convenience, denote $I=\{i \in [n] \mid y_i>x_i\}$ and $J= \{j \in [n] \mid x_j > y_j\}$.
We have $A = \{y_i-x_i \mid i\in I\}$ and $B=\{x_j-y_j \mid j \in J\}$ following (\ref{eqn:diff1}) and (\ref{eqn:diff2}).

As $(A,B)$ is reducible,
  there exist nonempty sets $I_0\subseteq I, J_0\subseteq J$ such that $\sum_{i\in I_0}(y_i-x_i)=\sum_{j\in J_0}(x_j-y_j)$, which implies that
\begin{equation*}
  \sum_{i\in I_0}y_i + \sum_{j\in J_0}y_j = \sum_{j\in J_0}x_j + \sum_{i\in I_0}x_i.
\end{equation*}

Note that $I_0\cap J_0=\emptyset$ because $I\cap J=\emptyset$. We further obtain
\begin{equation}\label{eq:adjust-for-m}
  \sum_{i\in I_0\cup J_0}y_i = \sum_{i\in I_0\cup J_0}x_i.
\end{equation}

We claim that $\sum_{i\in I_0\cup J_0}{R_i(y_i)}=\sum_{i\in I_0\cup J_0}{R_i(x_i)}$. The proof is as follows.

If $\sum_{i\in I_0\cup J_0}{R_i(y_i)}<\sum_{i\in I_0\cup J_0}{R_i(x_i)}$,
  we can see $\y$ is not $(k+1)$-optimal because by setting $y_i=x_i$ for $i\in I_0\cup J_0$, the revenue of $\y$ is enlarged.
Similarly, if $\sum_{i\in I_0\cup J_0}{R_i(y_i)}>\sum_{i\in I_0\cup J_0}{R_i(x_i)}$,
  we can see $\x$ is not $k$-optimal because by setting $x_i=y_i$ for $i\in I_0\cup J_0$, the revenue of $\x$ is enlarged.
Therefore, it must hold that $\sum_{i\in I_0\cup J_0}{R_i(y_i)}=\sum_{i\in I_0\cup J_0}{R_i(x_i)}$.

\smallskip Following the claim above,
  the revenue of $\y$ is unchanged (and hence $\y$ is still optimal) if we modify $y_i=x_i$ for all $i\in I_0\cup J_0$. 
    Note that such a modification of $\y$ would decrease $\sum A$, and  
     moreover $\sum A=\sum B + 1$ is always positive, therefore eventually $\y$ cannot be modified. 
    This means that $\diff(x,y)$ becomes irreducible after several modifications of $\y$. So the lemma holds. \qed
\end{proof}

As a side note, \Cref{lemma:diff-irreducible} implies that
\emph{for any optimal $k$-profile $\x$, where $k<nm$,
  there exists an optimal $(k+1)$-profile $\y$ such that $\sum_i |y_i-x_i|<2\lambda_m$.}

To see this, first find the optimal $(k+1)$-profile $\y$ with $\diff(\x,\y)=(A,B)$ irreducible.
Observe that $\sum B < \lambda_m$.
   Otherwise, $\sum B\geq \lambda_m$ and $\sum A\geq \lambda_m$, and $(A,B)$ is reducible by \Cref{lemma:property-for-m}.
 Therefore $\sum_i |y_i-x_i|=\sum B + \sum A < 2\lambda_m$.

\section{Algorithm for finding optimal $k$-profile}\label{sect:mainalgorithm}

It is sufficient to solving the following subproblem (for $k$ from $0$ to $nm-1$):

\begin{problem}\label{prob:sub}
Given a $k$-profile $\x^{(k)}$.
Among all the $(k+1)$-profile $\mathbf{y}$ with $\diff(\x,\y)$ being irreducible,
  find the one, denoted by $\x^{(k+1)}$, with the largest revenue.
\end{problem}

Clearly, we can set $\x^{(0)}$ to be the unique (and optimal) $0$-profile.
 Then, by induction, $\x^{(1)}$, $\ldots$, $\x^{(nm)}$ would all be optimal according to \Cref{lemma:diff-irreducible}.

In what follows we solve this subproblem in $O(f(m) \log n)$ time, where $f(m)$ is some function of $m$,
  and factor $\log n$ comes from the application of heap.

\newcommand{\DO}{\mathsf{DO}}
\newcommand{\UNDO}{\mathsf{UNDO}}

For convenience, assume $\x^{(k)}=(x_1,\ldots,x_n)$.

\subsubsection*{Data structures.}
Our algorithm uses $2m$ heaps. 

For each $d\in [m]$, we build a max-heap $\DO_d$ whose items are 
       those projects $i\in[n]$ for which $x_i+d\leq m$,
         and the value of item $i$ is defined by $R_i(x_i+d)-R_i(x_i)$ -- 
            the increase of revenue when we distribute $d$ more efforts into project $i$.
\begin{equation}
\begin{split}
  \DO_d = \{\langle i,R_i(x_i+d)-R_i(x_i)\rangle \mid i\in[n], x_i+d\leq m\}.
\end{split}
\end{equation}
 
For each $d\in [m]$, we build a min-heap $\UNDO_d$ whose items are 
       those projects $i\in[n]$ for which $x_i-d\geq 0$,
         and the value of item $i$ is defined by $R_i(x_i)-R_i(x_i-d)$ -- 
            the lost of revenue when we withdraw $d$ efforts from project $i$.
\begin{equation}
\begin{split}
  \UNDO_d = \{\langle i,R_i(x_i)-R_i(x_i-d)\rangle\mid i\in [n], x_i-d\geq 0\}.
\end{split}
\end{equation}

Observe that if $x_i$ is changed,
   we shall update the value of item $i$ (calling UPDATE\_VALUE) in each of the $2m$ heaps.
  (To be more clear, sometimes we may have to call DELETE or INSERT instead of UPDATE\_VALUE,
    since the condition $x_i+d\leq m$ may change, so as $x_i-d\geq 0$ after the change of $x_i$.)

\subsection{The algorithm}

Consider all irreducible pairs $(A,B)$ with $\sum A-\sum B=1$. Recall \Cref{example:ip2,example:ip3}.
For convenience, denote them by $(A_1,B_1),\ldots, (A_p,B_p)$, where $p$ is the number of such pairs.
(Note: we can generate and store these $p$ pairs by a brute-force preprocessing procedure, whose running time is only related to $m$. )

For each $c \in [p]$,
  denote by $\y^{(c)}$  the best $(k+1)$-profile among those satisfying $\diff(\x^{(k)},\y)=(A_c,B_c)$.
By Lemma~\ref{lemma:diff-irreducible}, the best among $\y^{(1)},\ldots,\y^{(p)}$ can serve as $\x^{(k+1)}$.

\medskip How do we compute $\y^{(c)}$ efficiently?

Let us first consider a simple case, e.g., $m=2$ and $(A_c,B_c)=(\{2\},\{1\})$.
In this case computing $\y^{(c)}$ is equivalent to solving the following problem:
    
Find the indices $i$ and $j$ that maximize 
    $$R_i(x_i+2)-R_i(x_i) - (R_j(x_j) - R_j(x_j-1)),$$
subject to  $$i\neq j, x_i+2\leq m, x_j-1\geq 0.$$

We can find $i$ so that $R_i(x_i+2)-R_i(x_i)$ is maximized using heap $\DO_2$,
 and find $j$ so that $R_j(x_j)-R_j(x_j-1)$ is minimized using heap $\UNDO_1$.
Clearly, $i\neq j$ because $x_i=0$ whereas $x_j>0$, and so the problem is solved.

\medskip Next, let us consider a more involved case: $m=3$ and $(A_c,B_c)=(\{2\},\{1\})$.
If we do the same as in the above case, it might occur that $i=j$ (for those $x_i=1$,
  item $i$ is in $\DO_2$ and $\UNDO_1$ simultaneously when $m=3$).

Nevertheless, utilizing the heaps, the above maximization problem can still be solved efficiently:
Find the best $i_1$ and second best $i_2$ in $\DO_2$, the best $j_1$ and second best $j_2$ in $\UNDO_1$,
  and moreover, try every combination $(i,j)\in \{(i_1,j_1),(i_1,j_2),(i_2,j_1),(i_2,j_2)\}$.
One of them must be the answer.
(Indeed, we can exclude $(i_2,j_2)$ from the trying set.)

\begin{figure}
  \centering
  \includegraphics[width=.8\textwidth]{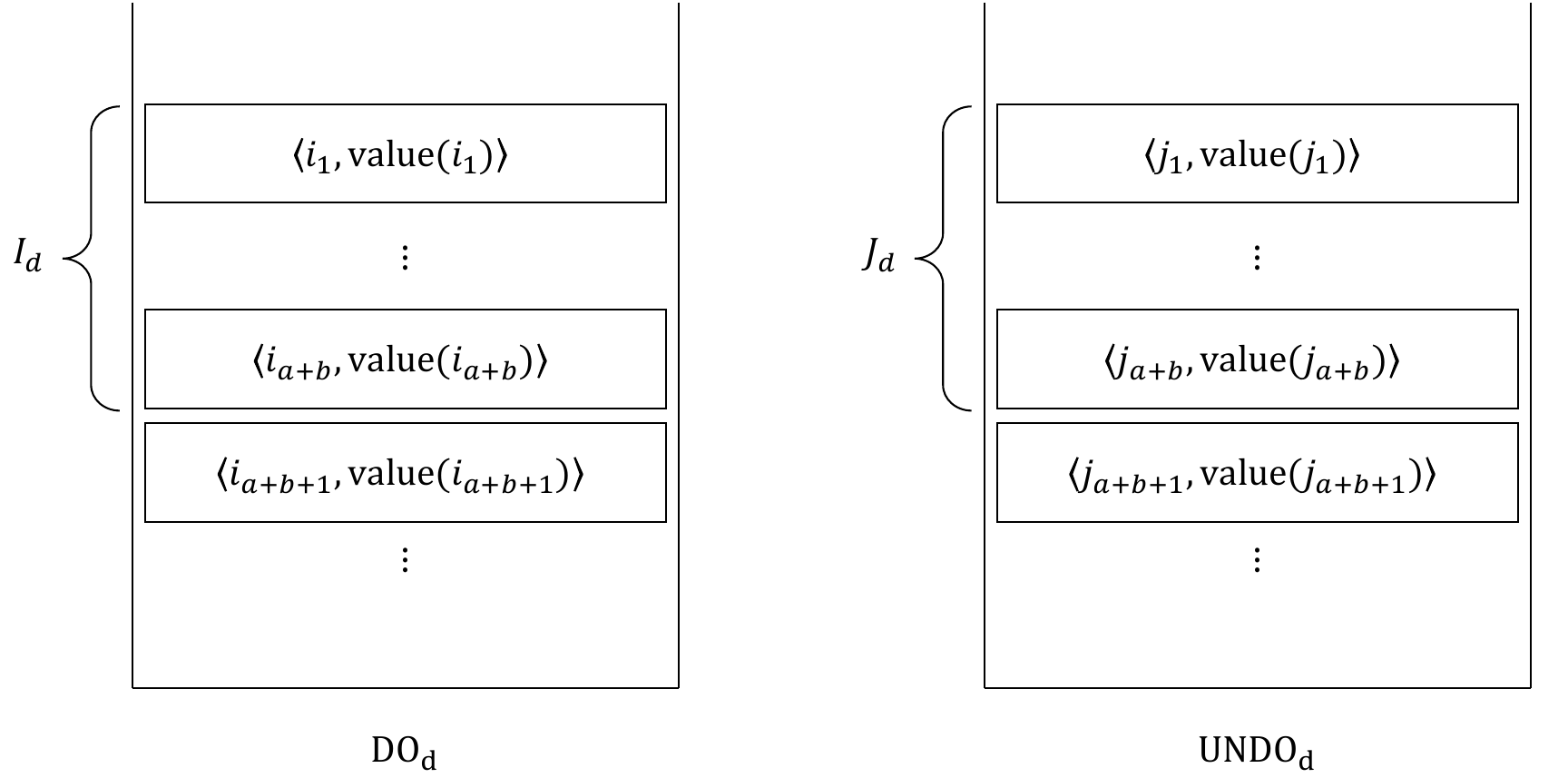}
  \caption{An illustration of the algorithm.}\label{fig:1}
\end{figure}

\medskip With the experience on small cases, we now move on to the general case.
   For $d\in [m]$, let $a_d$ denote the multiplicity of $d$ in $A_c$, and $b_d$ the multiplicity of $d$ in $B_c$.
    Let $a=\sum a_d$ and $b=\sum b_d$ be the number of elements in $A_c$ and $B_c$, respectively
    (which are bounded by $\lambda_m$ according to the analysis in \Cref{subsect:ip-brute-force}).

We now use a brute-force method to compute $\y^{(c)}$.   

1. For each $d\in [m]$, compute the set $I_d$ that contains the best $a+b$ items in $\DO_d$,
 and compute $J_d$ that contains the best $a+b$ items in $\UNDO_d$. See Figure~\ref{fig:1}.

2. Enumerate $(I'_1,\ldots,I'_m),(J'_1,\ldots,J'_m)$ such that 
    $$\left\{
        \begin{array}{ll}
          \text{$I'_d\subseteq I_d$ and $I'_d$ has size $a_d$,}\\
          \text{$J'_d\subseteq J_d$ and $J'_d$ has size $b_d$.}
        \end{array}
      \right.$$
   
    \quad When $I'_1,\ldots,I'_m,J'_1,\ldots,J'_m$ are pairwise-disjoint,
       we obtain a solution:  
    $$
        \text{increase $x_i$ by $d$ for $i \in I'_d$, and decrease $x_j$ by $d$ for $j \in J'_d$.}
    $$
    Select the best solution and it is $\y^{(c)}$.

The enumeration to compute an $\y^{(c)}$ takes $O\left((a+b)m\log n+g(m)\right)$ time, where
$$g(m)=O\left(\binom{a+b}{a_1}\dots\binom{a+b}{a_m}\binom{a+b}{b_1}\dots\binom{a+b}{b_m}\right).$$
So Problem~\ref{prob:sub} can be solved in $O\left(p(a+b)m\log n+pg(m)\right)$ time, recall that $p$ is the number of irreducible pairs $\left(A,B\right)$ satisfying $\sum A-\sum B=1$, entirely determined by $m$.

\section{Separable convex objective function}\label{sect:convex}

In this section, we consider a special case where 
   all the separated objective function $R_j$ are convex (the concave case has been studied extensively as mentioned in the introduction).
We present two algorithms for this special case:
  One runs in $O(nm + n\log n)$ time and it finds the optimal $k$-profile for all $k$.
  The other runs in $O(nm)$ time and it finds the optimal solution for a given $k$.

\begin{remark}
If $m$ is a constant, our first algorithm in this section runs in $O(n\log n)$ time, as the algorithm shown in \Cref{sect:mainalgorithm}.
However, the constant factor of the algorithm in this section is much smaller.
\end{remark}

\begin{lemma} \label{lemma:Only one not full}
There exists an optimal $k$-profile satisfies: At most one project receives more than $0$ and less than $m$ efforts.
\end{lemma}

\begin{proof}
  We prove it by contradiction. Suppose $\x = (x_1,\ldots,x_n)$ is an optimal $k$-profile.
    Assume $0<x_i<m$, $0<x_j<m$ for some $i\neq j$. 
    
  Assume $R_i(x_i+1)-R_i(x_i)\geq R_j(x_j+1)-R_j(x_j)$; otherwise we swap $i$ and $j$.
    We can withdraw one effort from project $j$ and give it to project $i$ without decreasing the total revenue. 
  By convexity of $R_i$ and $R_j$, the inequality $R_i(x_i+1)-R_i(x_i)\geq R_j(x_j+1)-R_j(x_j)$ still holds after such an adjustment, 
     so this process can be repeated until $x_i=m$ or $x_j=0$.
\end{proof}

First we sort the projects by $R_j(m)$ in descending order in $O(n\log n)$ time. 
Following \Cref{lemma:Only one not full}, when $k \bmod m = 0$, the largest revenue equals $\sum_{j=1}^{k/m} R_j(m)$ (trivial proof omitted).
Assume $k \bmod m\neq 0$ in the following.
In this case, there must one project that is allocated with $k \bmod m$ efforts. 

Denote $q(k) = \lfloor \frac{k}{m} \rfloor$. 

For $i\leq q(k)+1$, denote by $\x^{(i)}$ the profile that allocates $k\bmod m$ efforts to project $i$,
   and $m$ efforts to each project $j\in [q+1]\setminus \{i\}$.

For $i>q(k)$, denote by $\x^{(i)}$ the profile that allocates $k\bmod m$ efforts to project $i$,
   and $m$ efforts to each project  $j\in [q]$.
   
\begin{lemma} \label{lemma:Two type of optimal}
   One of $\x^{(1)},\ldots,\x^{(n)}$ is optimal. (It is trivial. Proof omitted.)
\end{lemma}

Denote by $Ans[k]$ the largest revenue of $k$-profile.

Denote by $Ans_1[k]= \max( \text{revenue of } x^{(i)} : i\leq q(k)+1\}$.

Denote by $Ans_2[k]= \max( \text{revenue of } x^{(i)} : i> q(k)+1\}$.

It follows from \cref{lemma:Two type of optimal} that $Ans[k] = \max (Ans_1[k], Ans_2[k])$.

We show how we compute the array $Ans_1$ altogether in $O(nm)$ time in the following.
The array $Ans_2$ can also be computed in $O(nm)$ time using similar idea (details omitted).

For each $c\in [m]\setminus 0$, 
 we compute $Ans_1[k]$ for $k$ congruent to $c$ (modulo $m$) in $O(n)$ time as follows,
    and thus obtain $Ans_1$ in $O(n m)$ time.

Define $r_j = R_j(m)-R_j(c)$. According to the definition of $x^{(i)}$ for $i\leq q(k)+1$,
\begin{equation}
    Ans_1[k] = \sum_{j=1}^{q(k)+1} R_j(m) - \min(r_1,\ldots, r_{q(k)+1}),
\end{equation}

As $k$ increases by $m$, quotient $q(k)$ increases by $1$, and we can compute 
  the term $\sum_{j=1}^{q(k)+1} R_j(m)$ and $\min(r_1,\ldots, r_{q(k)+1})$ both in $O(1)$ time.
Therefore it takes $O(1)$ time for each $k$ congruent to $c$.

\bigskip Now we move on the problem that asks $Ans[k]$ for a certain $k$.

In this problem, we do not have to sort $R_j(m)$. Instead,
  we only need to find out the largest $q(k)+1$ items of $R_1(m),\ldots, R_n(m)$,
     which takes $O(n)$ time through the algorithm for finding the $K$-th largest number in an array.
     Therefore, we cut off the term $O(n\log n)$ for this easier problem.

\section{An alternative algorithm when $m=2$}\label{sect:faster}

In this section, we describe an algorithm for $m=2$
  which costs $O(n)$ time after sorting.
  It also solves the problem for all $k~(0\leq k\leq mn)$.
 
\subsection{Preliminaries: some observations on (max,+) convolution}
\begin{definition}\label{def:oscillating-concave}
Given a function $g:[n]\rightarrow \mathbb{R}$, we say $g$ is \emph{oscillating concave}, if it satisfies the following properties:\\
For any $k$,\\
(1) $g(2k)-g(2k-2)\geq g(2k+2)-g(2k)$ (namely, $g(2k)$ is concave);\\
(2) $g(2k+2)-g(2k+1)\geq g(2k+1)-g(2k)$; \\
(3) $g(2k+1)-g(2k)\leq g(2k)-g(2k-1)$; \\
(4) $g(2k+1)-g(2k)$ is decreasing for $k$; \\
(5) $g(2k)-g(2k-1)$ is decreasing for $k$.
\end{definition}

\begin{lemma}\label{lemma:max-plus-convolution}
Let $f:[n]\rightarrow \mathbb{R}$ be a concave function and $g:[n]\rightarrow \mathbb{R}$ an oscillating concave function. The (max,+) convolution of $f$ and $g$:
$$h(k)=\max_{1\leq i\leq k}\left(f(i)+g(k-i)\right), 1\leq k\leq n,$$
can be computed in $O(n)$ time.
\end{lemma}

\begin{proof}
We demonstrate that:\\

\noindent{Observation 1. For any fixed $k$ ($1\leq k\leq n-1$), let $i_k$ be a maximum point of $f(i)+g(k-i)$. Then $f(i)+g(k+1-i)$ attains its maximum at a point in $\left\{i_k-1,i_k,i_k+1\right\}$.}\\

  The above observation indicates that we can compute $h(k+1)$ by $h(k)$ in $O(1)$ time. We prove it by contradiction. Let $i_{k+1}$ denote a maximum point of $f(i)+g(k-i)$, either $i_{k+1}>i_k+1$, or $i_{k+1}<i_k-1$.

  By the optimality of $i_k$ and $i_{k+1}$, we derive
  \begin{equation}\label{eq:maximal-point-of-k-(1)}
    f(i_k)+g(k-i_k)\geq f(i_{k+1}-1)+g(k-i_{k+1}+1)
  \end{equation}
  and
  \begin{equation}\label{eq:maximal-point-of-k+1-(1)}
    f(i_{k+1})+g(k+1-i_{k+1})> f(i_k+1)+g(k-i_k).
  \end{equation}
  Notice that the equation doesn't hold in (\ref{eq:maximal-point-of-k+1-(1)}) because $i_k+1$ is not a maximum point of $f(i)+g(k+1-i)$.
  
  Combining (\ref{eq:maximal-point-of-k-(1)}) and (\ref{eq:maximal-point-of-k+1-(1)}) we have
  $f(i_{k+1})-f(i_{k+1}-1)>f(i_k+1)-f(i_k).$
   By the concavity of $f$, this implies $i_{k+1}< i_k+1$, which contradicts to $i_{k+1}>i_k+1$. So we can assume $i_{k+1}<i_k-1$.
  
  By the optimality of $i_k$ and $i_{k+1}$, we derive
  \begin{equation}\label{eq:maximal-point-of-k-(2)}
    f(i_k)+g(k-i_k)\geq f(i_{k+1})+g(k-i_{k+1}),
  \end{equation}  
  and
  \begin{equation}\label{eq:maximal-point-of-k+1-(2)}
    f(i_{k+1})+g(k+1-i_{k+1})> f(i_k)+g(k+1-i_k).
  \end{equation}  
  Notice that the equation doesn't hold in (\ref{eq:maximal-point-of-k+1-(2)}) because $i_k$ is not a maximum point of $f(i)+g(k+1-i)$.
  
  Combining  (\ref{eq:maximal-point-of-k-(2)}) and (\ref{eq:maximal-point-of-k+1-(2)}) we derive
  \begin{equation}\label{eq:ineq-for-g}
    g(k+1-i_{k+1})-g(k-i_{k+1})>g(k+1-i_k)-g(k-i_k).
  \end{equation}
  
  Similarly, by the optimality of $i_k$ and $i_{k+1}$, we derive
  \begin{equation}\label{eq:maximal-point-of-k-(2)-2}
    f(i_k)+g(k-i_k)\geq f(i_{k+1}+1)+g(k-i_{k+1}-1).
  \end{equation}
  and
  \begin{equation}\label{eq:maximal-point-of-k+1-(2)-2}
    f(i_{k+1})+g(k+1-i_{k+1})> f(i_k-1)+g(k+2-i_k).
  \end{equation}
  Notice that the equation in (\ref{eq:maximal-point-of-k+1-(2)-2}) doesn't hold because $i_k-1$ is not a maximum point of $f(i)+g(k+1-i)$.
  
  Combine (\ref{eq:maximal-point-of-k-(2)-2}) and (\ref{eq:maximal-point-of-k+1-(2)-2}) together we derive
  \begin{align}\label{eq:add}
   f(i_k)-f(i_k-1)+&g(k-i_k)-g(k+2-i_k)> \nonumber \\
   &f(i_{k+1}+1)-f(i_{k+1})+g(k-i_{k+1}-1)-g(k+1-i_{k+1}).
  \end{align}
  By the concavity of $f$ and $i_{k+1}<i_k-1$, we have $f(i_k)-f(i_k-1)<f(i_{k+1}+1)-f(i_{k+1})$. Further by (\ref{eq:add}) we have
  \begin{equation}\label{eq:ineq-for-g-2}
    g(k-i_k)-g(k+2-i_k)>g(k-i_{k+1}-1)-g(k+1-i_{k+1}).
  \end{equation}
  
    We will use (\ref{eq:ineq-for-g}) and (\ref{eq:ineq-for-g-2}) to derive contradiction. For convenience, let $x=k-i_k, y=k+1-i_{k+1}$, and (\ref{eq:ineq-for-g}) can be simplified as 
    \begin{equation}\label{eq:simplified-ineq-for-g}
      g(y)-g(y-1)>g(x+1)-g(x),
    \end{equation}
    (\ref{eq:ineq-for-g-2}) can be simplified as
    \begin{equation}\label{eq:simplified-ineq-for-g-2}
      g(y)-g(y-2)>g(x+2)-g(x).
    \end{equation}
    By $i_{k+1}<i_k-1$ we know $y>x+2$. 
    \setcounter{case}{0}
    \begin{case}[$x,y$ are both even]
    By (\ref{eq:simplified-ineq-for-g-2}) and Definition~\ref{def:oscillating-concave}.1, we know $y\leq x+2$, which leads to a contradiction.
    \end{case}
    \begin{case}[$x,y$ are both odd]
    By Definition~\ref{def:oscillating-concave}.2, we have
    \begin{equation}\label{eq:case2-1}
      g(x+1)-g(x)\geq \frac{g(x+1)-g(x-1)}{2}.
    \end{equation}
    and
    \begin{equation}\label{eq:case2-2}
      \frac{g(y+1)-g(y-1)}{2}\geq g(y)-g(y-1).
    \end{equation}
    By Definition~\ref{def:oscillating-concave}.1 and $y\geq x+3$ we have
    \begin{equation}\label{eq:case2-3}
      \frac{g(x+1)-g(x-1)}{2}\geq \frac{g(y+1)-g(y-1)}{2}.
    \end{equation}
    Combine (\ref{eq:case2-1}), (\ref{eq:case2-2}) and (\ref{eq:case2-3}) we derive $g(x+1)-g(x)\geq g(y)-g(y-1)$, which contradicts to (\ref{eq:simplified-ineq-for-g}).
    \end{case}
    \begin{case}[$x$ is even, $y$ is odd]
    By Definition~\ref{def:oscillating-concave}.4 and $y\geq x+2$ we have $g(y)-g(y-1)\leq g(x+1)-g(x)$, which contradicts to (\ref{eq:simplified-ineq-for-g}).
    \end{case}
    \begin{case}[$x$ is odd, $y$ is even]
    By Definition~\ref{def:oscillating-concave}.5 and $y\geq x+2$ we have $g(y)-g(y-1)\leq g(x+1)-g(x)$, which contradicts to (\ref{eq:simplified-ineq-for-g}).
    \end{case}  \qed
\end{proof}
\subsection{Algorithm for finding optimal solutions based on (max,+) convolution}
Suppose the projects are sorted by $R_j(2)$ in descending order.
For convenience, let $a_j=R_j(1)$, and $b_j=R_j(2)-R_j(1)$.

Divide all projects into two groups $A,B$. Group $A=\{j\mid a_j>b_j\}$, and group $B=\{j\mid a_j\leq b_j\}$. The number of elements in $A$ is denoted as $|A|$, and $|B|$ analogously.

\begin{definition}
The maximal revenue of allocating $k$ efforts to group $A$ projects is denoted as $f(k)$. \\
The maximal revenue of allocating $k$ efforts to group $B$  projects is denoted as $g(k)$.
\end{definition}

The following lemma indicates how to compute $f$ and $g$.
\begin{lemma}[Calculate $\boldsymbol{f,g}$]\label{lemma:calculate-fg}
\begin{enumerate}
  \item $f(k)=$ sum of the kth largest $a_i,b_i$, where $i\in A$. 
  \item $g(k)=\left\{
            \begin{aligned}
            &\sum_{i=1}^{\frac{k}{2}}R_i(2) , k\text{ is even, } \\
            &\max\left(g(k-1)+\max_{\frac{k+3}{2}\leq i\leq |B|}a_i, g(k+1)-\min_{1\leq i\leq \frac{k+1}{2}}b_i\right) , k\text{ is odd, } \\
            \end{aligned}
            \right.
        $
  
  where $i\in B$.
\end{enumerate}   
\end{lemma}
\begin{proof}
  1. Proof is evident.
  
  2. Proof is evident when $k$ is even.
  
  When $k$ is odd, we demonstrate that for projects in group $B$, there exists an optimal $k$-profile, such that a unique project is allocated with one effort. 
  
  We prove it by contradiction. Assume there are two projects $i,j\in B$ receiving one effort separately. Without loss of generality, suppose $a_i\leq a_j$, then we 
  have $a_i\leq a_j\leq b_j$. We can remove one effort from $i$-th project and allocate it to $j$-th project, without decreasing the total revenue.  \qed

\end{proof}

Denote the maximal revenue of allocating $k$ efforts to all projects as $h(k)$, then $h(k)$ can be written as (max,+) convolution of $f$ and $g$ as follows:
\begin{equation*}\label{eq:calculate-ANS}
  h(k)=\max_{0\leq i\leq 2|A|, 0\leq k-i\leq 2|B|}f(i)+g(k-i).
\end{equation*}

The following lemma together with Lemma~\ref{lemma:max-plus-convolution} ensure that we can compute $h(k) (1\leq k\leq 2n)$ in $O(n)$ time.
\begin{lemma}[Properties of $\boldsymbol{f,g}$]\label{lemma:properties-of-fg}
\begin{enumerate}
(1) $f(k)$ is an convex function. \\
(2) $g(k)$ is an oscillating concave function. 
\end{enumerate}
\end{lemma}
\begin{proof}
1. Proof is evident by Lemma~\ref{lemma:calculate-fg}.1.

2. By Lemma~\ref{lemma:calculate-fg}.2., $g(2k)$ is concave. We only need to prove $g$ satisfies oscillating concave property (2),(3),(4) and (5) in Definition~\ref{def:oscillating-concave}.

\bigskip

Prove Property (2):

It's equivalent to proving $g(2k+2)+g(2k)\geq 2g(2k+1)$. By Lemma~\ref{lemma:calculate-fg}.2 we know
$$g(2k+2)+g(2k)=\sum_{i=1}^{k+1}(a_i+b_i)+\sum_{i=1}^{k}(a_i+b_i),$$
and
$$2g(2k+1)=2\max{\left(\sum_{i=1}^{k}(a_i+b_i)+\max_{k+2\leq i\leq |B|}a_i, \sum_{i=1}^{k+1}(a_i+b_i)-\min_{1\leq i\leq k+1}b_i\right)},$$
where $a_i\leq b_i$.

It reduces to prove 
\begin{equation}\label{eq:eq1}
  \sum_{i=1}^{k+1}(a_i+b_i)+\sum_{i=1}^{k}(a_i+b_i)\geq 2\left(\sum_{i=1}^{k}(a_i+b_i)+\max_{k+2\leq i\leq |B|}a_i\right),
\end{equation}
and
\begin{equation}\label{eq:eq2}
  \sum_{i=1}^{k+1}(a_i+b_i)+\sum_{i=1}^{k}(a_i+b_i)\geq 2\left(\sum_{i=1}^{k+1}(a_i+b_i)-\min_{1\leq i\leq k+1}b_i\right).
\end{equation}
First we prove (\ref{eq:eq1}). It can be simplified as 
$$a_{k+1}+b_{k+1}\geq 2\max_{k+2\leq i\leq |B|}a_i.$$

Let $a_{i_0}=\max_{k+2\leq i\leq |B|}a_i$, we know $a_{k+1}+b_{k+1}\geq a_{i_0}+b_{i_0}\geq a_{i_0}+a_{i_0}$. Therefore $a_{k+1}+b_{k+1}\geq 2a_{i_0}$.

Next we prove (\ref{eq:eq2}). It can be simplified as
$$2\min_{1\leq i\leq k+1}b_i\geq a_{k+1}+b_{k+1}.$$
Let $b_{i_1}=\min_{1\leq i\leq k+1}b_i$, we know $a_{k+1}+b_{k+1}\leq a_{i_1}+b_{i_1}\leq b_{i_1}+b_{i_1}$. Therefore $2b_{i_1}\geq a_{k+1}+b_{k+1}$.

\bigskip

Prove Property (3):

By Lemma~\ref{lemma:properties-of-fg}.2 we have
$$g(2k+2)-g(2k)\leq g(2k)-g(2k-2),$$
by Lemma~\ref{lemma:properties-of-fg}.3 we have
$$g(2k+1)-g(2k)\leq \frac{g(2k+2)-g(2k)}{2},$$
and
$$\frac{g(2k)-g(2k)-2}{2}\leq g(2k)-g(2k-1).$$
Combine the above three inequalities we can get $g(2k+1)-g(2k)\leq g(2k)-g(2k-1)$.

\bigskip

Prove Property (4):

Formally we need to prove
$$g(2k-1)-g(2k-2)\geq g(2k+1)-g(2k).$$ 
By Lemma~\ref{lemma:calculate-fg}.2, we have
\begin{equation*}
  g(2k-1)-g(2k-2) = \max\left( \max_{k+1\leq i\leq |B|}a_i, a_{k}+b_{k}-\min_{1\leq i\leq k}b_i \right),
\end{equation*}
and 
\begin{equation}\label{eq:5-2}
  g(2k+1)-g(2k) = \max\left( \max_{k+2\leq i\leq |B|}a_i, a_{k+1}+b_{k+1}-\min_{1\leq i\leq k+1}b_i \right).
\end{equation}

If $\min_{1\leq i\leq k+1}b_i\ne b_{k+1}$, then $\min_{1\leq i\leq k+1}b_i=\min_{1\leq i\leq k}b_i$. By $\max_{k+1\leq i\leq |B|}a_i\geq \max_{k+2\leq i\leq |B|}a_i$, 
and $a_k+b_k\geq a_{k+1}+b_{k+1}$, we know $g(2k-1)-g(2k-2)\geq g(2k+1)-g(2k)$.

Otherwise, $\min_{1\leq i\leq k+1}b_i= b_{k+1}$. Then (\ref{eq:5-2}) can be simplified as
\begin{align*}
  g(2k+1)-g(2k) &= \max\left(\max_{k+2\leq i\leq |B|}a_i, a_{k+1}\right), \\
   &= \max_{k+1\leq i\leq |B|}a_i. 
\end{align*} 
So $g(2k-1)-g(2k-2)\geq g(2k+1)-g(2k)$. 

\smallskip

Prove Property (5):

Formally we need to prove 
\begin{equation*}
  g(2k)-g(2k-1)\geq g(2k+2)-g(2k+1).
\end{equation*}
By Lemma~\ref{lemma:calculate-fg}.2, we have
\begin{equation*}
  g(2k)-g(2k-1)=\min\left(a_k+b_k-\max_{k+1\leq i\leq |B|}a_i, \min_{1\leq i\leq k}b_i\right).
\end{equation*}
and
\begin{equation}\label{eq:6-2}
  g(2k+2)-g(2k+1)=\min\left(a_{k+1}+b_{k+1}-\max_{k+2\leq i\leq |B|}a_i, \min_{1\leq i\leq k+1}b_i\right).
\end{equation}

If $\max_{k+1\leq i\leq |B|}a_i\ne a_{k+1}$, then $\max_{k+1\leq i\leq |B|}a_i=\max_{k+2\leq i\leq |B|}a_i$. By $a_k+b_k\geq a_{k+1}+b_{k+1}$, and $\min_{1\leq i \leq 
k}b_i\geq \min_{1\leq i\leq k+1}b_i$, we know $g(k)-g(2k-1)\geq g(2k+2)-g(2k+1)$.

Otherwise, $\max_{k+1\leq i\leq |B|}a_i= a_{k+1}$, then $a_{k+1}\geq \max_{k+2\leq i\leq |B|}a_i$. So 
\begin{align*}
a_{k+1}+b_{k+1}-\max_{k+2\leq i\leq |B|}a_i &\geq b_{k+1}, \\
    &\geq \min_{1\leq i\leq k+1}b_i.
\end{align*}

So (\ref{eq:6-2}) can be simplified as
\begin{equation*}
  g(2k+2)-g(2k+1)=\min_{1\leq i\leq k+1}b_i.
\end{equation*}

Therefore $g(2k)-g(2k-1)\geq \min_{1\leq i\leq k}b_i\geq \min_{1\leq i\leq k+1}b_i=g(2k+2)-g(2k+1)$.   \qed
\end{proof}
\section{Summary}

We revisit the classic resource allocation problem with a separable objective function under a single linear constraint. 
A regret-enabled greedy algorithm is designed that achieves $O(n\log n)$ time for $m=O(1)$, 
  outperforming dynamic programming algorithm for small $m$. 
The new algorithm is practical especially for very small $m$,
  and its analysis is not over complicated (see Lemma~\ref{lemma:diff-irreducible}).

For the special case where 
   all the separated objective function $R_j$ are convex, we present fast algorithms 
  that cost $O(nm + n\log n)$ time (for all $k$) or $O(nm)$ time (for one given $k$).
For the special case where $m=2$,
  we show that the main algorithm only costs linear time $O(mn)=O(n)$, after a sorting process that costs $O(n \log n)$ time.
  It arises an open question what is the lower bound for this allocation problem (for $m=2$ or $m=O(1)$).
  
A more interesting open question (suggested by one reviewer) is that 
     can we solve this resource allocation problem in time $O(n\log n \cdot \mathsf{poly}(m))$ or even $O(n\log n + \mathsf{poly}(m))$?

\bibliographystyle{unsrt}
\bibliography{RGA}

\clearpage

\end{document}